\documentclass{article}
\usepackage{amsthm,amsmath,amsfonts}
\usepackage{authblk}
\usepackage[T1]{fontenc}
\newcommand{\keywords}[1]{\textbf{Keywords: } #1}
\newtheorem{definition}{Definition}
\newtheorem{proposition}{Proposition}
\newtheorem{theorem}{Theorem}

\newtheorem{lemma}{Lemma}
\newtheorem{corollary}{Corollary}
\title{Automatic Randomness Tests}
\author{Birzhan Moldagaliyev}
\affil{Department of Mathematics\\ National University of Singapore}
\date{}
\begin{document}

\maketitle

\begin{abstract}
\noindent
In this paper we define a notion of automatic randomness tests~(ART) which capture measure theoretic typicalness of infinite binary sequences within the framework of automata theory. An individual ART is found to be equivalent to a deterministic B\"{u}chi automaton recognizing $\omega$-language of (Lebesgue) measure zero. A collection of ART's induce a notion of automatic random sequence. We provide a purely combinatorial characterization of an automatic random sequence in the form of a disjunctive property for sequences. At last, we compare two kinds of automatic randomness tests presented in this paper.
\end{abstract}

\section{Introduction}
The theory of algorithmic randomness \cite{Downey} tries to explain what kind of properties make an individual element of a sample space to appear random. Mostly the theory deals with infinite binary sequences. One of the first works in the field could be attributed to Borel \cite{Borel} with his notion of normal numbers. An emergence of the computability theory allowed to formalize a notion of randomness with respect to certain classes of algorithms. Loosely speaking, an object is considered random, when any kind of algorithm fails to find sufficient patterns in the structure of that object. There are many kinds of randomness definitions within the theory, however most of them arise from three paradigms of randomness:
\begin{enumerate}
	\item Unpredictability
	\item Incompressibility
	\item Measure theoretical typicalness
\end{enumerate}
Each paradigm has its own tools to explore randomness. In relation to the unpredictability paradigm different kinds of effective martingales are considered \cite{Schnorr1971}. On the other hand, the incompressibility paradigm works with a notion of complexity, the most famous example being the Kolmogorov complexity \cite{Kolmogorov}. Finally, the measure theoretical typicalness paradigm tries to effectivize a notion of a nullset in measure-theoretic sense, with the most prominent example in the face of Martin-L\"{o}f randomness tests \cite{Lof}. An interesting part is that one could choose an appropriate definition from each paradigm so that arising classes of random infinite sequences coincide. This interconnection is thought to symbolize a universality of the notion. In almost all of these studies, a full computational strength of a Turing machine is used. On the one hand, it guarantess a generality of resulting theory, because any kind of effective process can be designed as an instance of a Turing machine. However, nothing prevents us from considering weaker variants of computational machines. For example, one could consider a polynomial-time computability, which gives rise to the notion of resource-bounded measure \cite{Lutz}. One can go even further and consider finite state machines, thus hitting the bottom of a computational machines hierarchy. Most paradigms for randomness have their variants adopted for finite state machines. As for the unpredictability paradigm, there are automatic martingales \cite{Stimm} and finite state predicting machines \cite{OConnor}. As for the incompressibility paradigm, one has several variants of automatic complexity for finite strings \cite{Shallit, Hyde, Calude2,Becher,Shen}. On top of that, there are approaches \cite{Busse} to study a notion of genericity from the computability theory, in the context of automata theory. From perspective of measure typicalness, we have a notion of regular nullset defined by Staiger \cite{Null, Monadic}. Those regular nullsets correspond to Buchi recognizable $\omega$-languages of measure zero. Our aim is to replicate the construction of Martin-L\"{o}f randomness tests, where nullsets are defined as intersection of classes having some algorithmic property, within automata theoretic framework. 

\section{Background}
In this section we review concepts on which the main part of the paper going to rely on. 
\subsection{Finite State Machines}
Let $S$ be a finite set, elements of which are to be interpreted as states. Consider a binary alphabet $\Sigma=\{0,1\}$. By concatentating elements of $\Sigma$, we obtain words, e.g. $00$, $101$, etc. We denote the collection of all words $\Sigma^*$, which is a monoid under concatenation and an identity element of the empty string, $\varepsilon$. Let $\Sigma$ act on $S$ from the right using a function $f: S \times \Sigma \rightarrow S$
and we can extend $f$ to an action $\cdot$ on the full monoid $\Sigma^*$ by defining
$$
s \cdot \varepsilon = s \mbox{ and }
s \cdot a = f(s,a) \mbox{ and } s \cdot (xy) = (s \cdot x) \cdot y
$$
for all states $s \in S$ and letters $a \in S$ and words $x,y \in \Sigma^*$. We choose some state, $s_0$, to be a \emph{starting state}. Furthermore, we are going to assume that all other states are reachable from $s_0$, i.e. for all $s\in S$, there is $w\in \Sigma^*$ such that $s_0\cdot w = s$. A \emph{Finite State Machine~(FSM)} is given as a triple of data: $(S,f,s_0)$. Given two words $x,y$ such that $x=yz$ for some word $z$, $y$ is said to be a \emph{prefix} of $x$, which is denoted as $y\preceq x$. In case $x\neq y$, it is said to be a \emph{strict prefix}, denoted as $y\prec x$. A term \emph{language} refers to any subcollection of $\Sigma^*$.
\paragraph*{Connected Components of FSM}
Suppose we are given a FSM, $M=(S,f,s_0)$. We can turn $S$ into a partially ordered set as follows:
\[
x\ge y\, \Leftrightarrow \,\exists w\in \Sigma^* \text{ such that } x\cdot w = y
\]
In other words, $x\ge y$ iff $y$ is reachable from $x$, for $x,y\in S$. The given partial order is reflexive and transitive, so it is possible to define following equivalence relation $\sim$ on $S$:
\[
x\sim y \, \Leftrightarrow x\ge y\text{ and } y\ge x
\]
Above equivalence relation tells that $x,y\in S$ belong to the same connected component. By taking quotient of $S$ under $\sim$ we end up with collection of connected components $[M]$ with inherited partial order. We call $g\in [M]$ a \emph{leaf connected component} if it is minimal, i.e. there is no element of $[M]$ strictly less than $g$. We denote a collection of leaf connected components of $M$ as $(M)$. 

\paragraph*{Finite Automata}
Let us enrich the structure of finite state machine, $(S,f,s_0)$, by introducing a set of \emph{accepting states} $F\subseteq S$. In order to simplify the notation, we reserve the letter $M$ for all of machines to be introduced. A \emph{(deterministic) finite automaton}, $M$, is given by a quadruple of data: $(S,f,s_0,F)$. A word $w\in \Sigma^*$ is said to be accepted by $M$ if:
\[
s_0\cdot w\in F
\]
Given a finite automaton $M$, the language of $M$, $L(M)$, corresponds to the collection of all words accepted by $M$. A language $L\subseteq \Sigma^*$ is said to be \emph{regular} if there is a finite automaton $M$ recognizing it, i.e. $L=L(M)$.

\paragraph*{Automatic Relation}
A membership of some word $w$ in a regular language $L$ can be easily checked on a corresponding finite automaton of $L$. Although this property is rather convenient, it is quite limiting in a sense that it is unary. There is a way to extend a notion of 'regularity' or 'automaticity' on relationships of arbitrary arity. Given some input of $k$ words $(w_1,w_2,\ldots,w_k)$ it is possible to convolute them into the one block in the following manner:
\[\left[
\begin{array}{ c c c c c c}
	( & \cdots & w_1 & \cdots & \#\# & )\\
	(& \cdots & w_2 & \cdots & \cdots & )\\
	(& \cdots & \cdots & \cdots & \cdots & )\\
	(& \cdots & w_k & \cdots & \# & )
\end{array}
\right]
\]
where length of the block is equal to the length of the longest word, while empty spaces at the end of shorter words are filled with a special symbol such as $\#$. Each column of such a block can be considered as a single symbol coming from a product space $\Sigma_{\#}^k$, where $\Sigma_{\#}$ refers to $\Sigma$ with added special symbol, i.e. $\Sigma_{\#}=\{0,1,\#\}$. By analogy with finite state machines, let $\Sigma_{\#}^k$ act on some finite set $S$ forming a function $f$:
\[
f:S\times \Sigma_{\#}^k\to S
\] 
This function is then extended to the action of $(\Sigma_{\#}^k)^*$ on $S$, defining a finite state machine $M=(S,f,s_0,F)$ with some $s_0\in S$ and $F\subseteq S$. Given block $u$ of above type is said to be accepted by $M$, if the run of $u$ on $M$ finishes at an accepting state, i.e. $s_0\cdot u\in F$. A relation $R\subseteq \Sigma^k$ is said to be automatic if there is a finite automaton $M$ which recognizes elements of $R$ given in the block form as above. A function $\phi:\Sigma^k \to \Sigma^m$ is said to be automatic if its graph forms an automatic relation, i.e. $graph(\phi)=\{(x,\phi(x))\mid x\in \Sigma^k\}\subseteq \Sigma^{k+m}$ is automatic. Let us give an example of automatic relation:\\
Let $R$ be a binary relation on words, $R\subseteq (\Sigma^*)^2$, such that $(x,y)\in R$ iff $\vert x \vert \le \vert y \vert$. We can construct a following finite automaton $M=(S,f,s_0,F)$ to recognize $R$. Firstly, we set $S=\{s_0,s_1\}$ with $F = \{s_0\}$. As for a function $f$, we set:
\begin{align*}
s_0\cdot\begin{pmatrix}a\\b\end{pmatrix} = \begin{cases} s_1, & \text{ if } a=\#\\s_0,& \text{ otherwise }\end{cases} 
&&
s_1\cdot\begin{pmatrix} a\\b \end{pmatrix} = s_1, \text{ for all } a,b
\end{align*}
Automatic relations enjoy quite convenient properties from logical point of view. This fact is captured by the following theorem by Khoussainov and Nerode \cite{Kho}:
\begin{proposition}
Let $R$ be a first-order definable relation on $(\Sigma^*)^k$ from given functions $(f_1,f_2,\ldots,f_n)$ and relations $(R_1,R_2,\ldots,R_m)$. If each of these functions and relations is automatic, then $R$ is also automatic. 
\end{proposition}

\paragraph*{Automatic Family}
In the computability theory, there is a notion of uniformly computably enumerable sets. This notion allows to look at the collection of sets from a single effective frame of reference. A somewhat parallel notion exists in automata theory known as \emph{Automatic family} \cite{Jain}.
\begin{definition}[Automatic family \cite{Jain}]
An automatic family is a collection of languages $\mathcal{U} = (U_i)_{i\in I}$ such that:
\begin{enumerate}
	\item $I$ is a regular set;
	\item $\{(x,i)\mid x\in U_i\}$ is an automatic relation.
\end{enumerate}
\end{definition}
	
\subsection{Infinite sequences}
While we refer to elements of $\Sigma^*$ as strings or words, we refer to an one-way infinite sequence on $\Sigma$ as a \emph{sequence} or $\omega$-\emph{word}. Alternatively, $\omega$-word $X$ can be a thought as a function:
\[X:\mathbb{N}\to \{0,1\}
\]
Given a sequence $X$, its elements can be indexed as $X=X_1X_2X_3\ldots$. Given a word $w$ and a sequence $X$, one can concatenate them, forming a sequence $Y=wX$. Herein, $w$ is called a \emph{prefix} of $Y$ and denoted as $w\prec Y$. If the length of $w$ is $i$, $\vert w \vert = i$, then $w$ is denoted as $Y[i]$. A collection of all prefixes of $Y$ is denoted as $Pref(Y)$. On the other hand $X$ is called \emph{suffix} of $Y$ and denoted as $X=Y[i:]$. A collection of all sequences, $\{0,1\}^{\mathbb{N}}$, is also called a \emph{Cantor Space}. A subcollection of the Cantor space is usually referred to as a \emph{class} and denoted with italicized capital letters. Given a language $L$ and a class $\mathcal{C}$, we can take their product:
\[L\cdot \mathcal{C} = \{wX\mid w\in L,\, X\in\mathcal{C}\}
\]

\paragraph*{B\"{u}chi Automata}
Automata which process sequences are usually called $\omega$-automata. The simplest of such automata are \emph{B\"{u}chi automata}. As we are going to work with deterministic B\"{u}chi automata, we limit ourselves revewing such automata only. However, reader should note that the term B\"{u}chi automata, generally, refers to nondeterministic B\"{u}chi automata. A deterministic B\"{u}chi automaton $M$ is given by a quadruple $M=(S,f,s_0,F)$ just as finite automata. However, an acceptance condition needs to be revised. Given a sequence $X$, it induces an infinite run of states, $S(X)$, in the following manner:
\begin{align*}
S(X)_1 = s_0 && S(X)_n \cdot X_n = S(X)_{n+1}, \text{ for all } n
\end{align*}
Let us denote the collection of elements of $S$ which appear infinitely often in $S(X)$ as $I(X)$. A sequence $X$ is accepted by $M$ iff $s\in F$ for some $s\in I(X)$, i.e. some accepting state appears infinitely often in $S(X)$. A collection all sequences accepted by $M$ forms the class $L(M)$. A class $\mathcal{C}$ is said to be recognized by determinstic B\"{u}chi automata if there $M$ as above such that $L(M)=\mathcal{C}$. A B\"{u}chi automaton $M$ is said to be of measure $m$ if $\mu(L(M))=m$, where $\mu$ is the Lebesgue measure on the Cantor Space. 

\paragraph*{Muller Automata}
To review a structure of deterministic Muller automata, we start with a finite state machine $(S,f,s_0)$. We enrich its structure by a subset of powerset of $S$, i.e. accepting collection $F\subseteq \mathcal{P}(S)$. A Muller automata $M$ is given by a quadruple: $M=(S,f,s_0,F)$. A sequence $X$ is said to be accepted by $M$ if 
\[
I(X)\in F
\]
Again, a collection of all sequences accepted by $M$ forms the language of $M$, $L(M)$. A class $\mathcal{C}$ is said ot be recognized by determinsitic Muller automata if there is $M$ as above such that $L(M)=\mathcal{C}$. A Muller automaton $M$ is said to be of measure $m$ if $\mu(L(M))=m$, where $\mu$ is the Lebesgue measure on the Cantor Space.

\paragraph*{Equivalence of B\"{u}chi and Muller automata}
It turns out that nondeterministic B\"{u}chi automata are equivalent to deterministic Muller automata in expressive power. This result is known as McNaughton's theorem \cite{Mc}. 

\subsection{Measure on the Cantor space}
The Cantor space can be given a Lebesgue measure, $\mu$, full details of which can be found in the book by Oxtoby \cite{Oxtoby}. For the sake of completeness, we are going to review essentials needed for upcoming discussions. A given measure $\mu$ can be thought of as a product Bernoulli measure induced by the equiprobable measure $\mu_0$ on $\{0,1\}$. Given a word $w$, let 
\[
[w] = \{w\}\cdot \{0,1\}^{\mathbb{N}}
\]
In a more general way, given a language $L$:
\[
[L]=L\cdot \{0,1\}^{\mathbb{N}}
\]
be the class of all sequences extending elements of $L$. Former classes form basic open classes with $\mu[w]=2^{-\vert w \vert}$. Given a language $W$, it is said to be an $\alpha$-cover of a class $\mathcal{C}$ if:
\[
\mathcal{C}\subseteq \bigcup_{w\in W}[w]\quad \text{and} \quad\sum_{w\in W}\mu[w] \le \alpha
\]
\begin{definition}
A class $\mathcal{C}$ is said to be of measure $0$, if there is $2^{-n}$-cover of $\mathcal{C}$, for all $n\ge 0$. A class $\mathcal{C}$ is said to be of measure $1$, if $\overline{\mathcal{C}}$ is of measure $0$.
\end{definition}

\paragraph*{Disjunctive sequences}
We are going to review a definition of \emph{disjunctive property} \cite{Calude} for sequences, which is going to be useful later on. 
\begin{definition}[Disjunctive sequence \cite{Calude}]
An infinite sequence $X$ is said to be disjunctive if any word $w\in \Sigma^*$ appears in $X$ as a subword. 
\end{definition}
\noindent
Let $\mathcal{D}$ denote the collection of all disjunctive sequences. It has been shown \cite{Calude} that $\mu(\mathcal{D})=1$. 

\section{Definitions}
\subsection{Main Definitions}
In this section we are going to define a notion of randomness tests in the context of automata theory. In doing so, we draw inspiration from the original definition of effective randomness tests by Martin-L\"{o}f \cite{Lof}. The main contribution of the mentioned paper is a process of defining measure zero classes in an algorithmic fashion. More precisely, one considers a uniformly recursively enumerable collection of languages,  $\mathcal{V}=(V_i)_{i\in \mathbb{N}}$, such that the measure of an individual class satisfies $\mu[V_i]\le 2^{-i}$. An effective nullset corresponding to the test $\mathcal{V}$ is given by $\bigcap_{i\in \mathbb{N}}[V_i]$. This shows that a definition of randomness tests requires only two concepts:
\begin{itemize}
	\item Uniform collection of languages.
	\item Condition on their measure.
\end{itemize}
In order to translate a concept of randomness tests into the domain of automata theory, we adopt a following strategy. As for a uniform collection of languages, we adopt a notion of an automatic family. As for a condition on their measure, we let measures of these sets be arbitrarily small. These ideas are captured in the following definitions. 
\begin{definition}[Martin-L\"{o}f automatic randomness tests]
Let $\mathcal{U}=(U_i)_{i\in I}$ be an automaic family. We say that $\mathcal{U}$ forms a Martin-L\"{o}f automatic randomness test~(MART) if:
\[
\mu[U_i]\le 2^{-\vert i \vert} \text{ for all }i\in I
\]
\end{definition}
\noindent
Above definition is a direct analog of Martin-L\"{o}f randomness tests, because any individual class has a condition on its measure. This is an example of a local condition on the measures. Instead, one could also have a global condition on the measures. This way, we obtain an analog of weak 2-randomness \cite{Downey}.
\begin{definition}[Automatic randomness tests]
Let $\mathcal{U}=(U_i)_{i\in I}$ be an automatic family. We say that $\mathcal{U}$ forms an Automatic randomness test~(ART) if 
\[
\liminf_{i\in I}\mu[U_i]=0
\]
\end{definition}
\noindent
Corresponding nullsets are defined in the manner of Martin-L\"{o}f randomness tests.
\begin{definition}[Covering]
Given an (M)ART $\mathcal{U}=(U_i)_{i\in I}$, let 
\[F(\mathcal{U})=\bigcap_{i\in I}[U_i]
\]
be its covering region. An infinite sequence $X$ is said to be covered by $\mathcal{U}$ if it belongs to the covering region of $\mathcal{U}$, i.e. $X\in F(\mathcal{U})$. A pair of (M)ART's $(\mathcal{U},\mathcal{V})$ is said to be equivalent if $F(\mathcal{U})=F(\mathcal{V})$. 
\end{definition}
\noindent
Finally, we define a notion of a random sequence in parallel with the original definition by Martin-L\"{o}f:
\begin{definition}[Random Sequence]
An infinite sequence $X$ is said to be \emph{Martin-L\"{o}f automatic random~(MAR)} if $X$ is not covered by any MART. Similarly, an infinite sequence $X$ is said to be \emph{automatic random~(AR)} if $X$ is not covered by any ART.
\end{definition}
\noindent
A theory of ART is more general and richer compared to that of MART. For this reason, we are going to focus our attention on ART's. Closer to the end of the paper, we are going to show that ART's are not equivalent to MART's. 

\subsection{Examples}
Let us provide some examples of automatic randomness tests. As groups are understood by their actions, ART's are understood by the type of sequences they cover. Let us construct an automatic randomness test covering some ultimately periodic infinite sequence $X$, i.e. $X=uv^{\omega}$ for some $u,v\in \Sigma^*$. To build covering ART, $\mathcal{U}=(U_i)_{i\in I}$, let us have:
\[
I = uv^*, \quad U_i = \{i\}
\]
In some sense, ultimately periodic sequences are too rigid, so it does not take a significant effort to come up with an ART covering it. To consider more complex examples, let us take a following class of infinite sequences:
\[\mathcal{C} = \{X\mid X_{2i}=0 \text{ for all }i\in \mathbb{N}\}
\]
In other words class $\mathcal{C}$ is a collection of sequences having $0$ at even numbered positions. In terms of the computability theory, each member, $X$, of $\mathcal{C}$ has a form $X=A\oplus 0^{\omega}$. To construct ART $\mathcal{U}=(U_i)_{i\in I}$ covering the class $\mathcal{C}$, let us take:
\[I=(00)^*,\quad U_i=(\Sigma 0)^{\frac{\vert i \vert}{2}}
\]
It is clear that $(U_i)_{i\in I}$ is a valid automatic family. Since  $\mu[U_i]=2^{-\frac{\vert i \vert}{2}}$, we have that $\lim_{i}\mu[U_i]=0$. Moreover $X\in \mathcal{C}$ if and only if $X\in [U_i]$ for all $i\in I$. Hence, $\mathcal{C}$ is exactly a covering region of $\mathcal{U}$. With reference to computability theory, Turing degree of $X=A\oplus 0^{\omega}$ is that of $A$. As $A$ can be chosen to be a sequence of any Turing degree, we have a following fact:

\begin{proposition}[Covering sequences of any complexity]
There is ART covering sequences of arbitrarily high Turing degree
\end{proposition}\noindent
This result might come as a quite unexpected, as ART covered sequences are thought to be of low computational complexity. 

\section{Properties}
In this section we study properties of an ART. Firstly, we would like to focus on immediate properties. 
\subsection{Immediate Properties}
We are going to show that an ART can be assumed to have a simple form.
\begin{theorem}[Assumptions on ART]
Let $\mathcal{U}$ be an ART. Then there is an equivalent ART $\mathcal{V}=(V_j)_{j\in J}$ satisfying following properties:
\begin{itemize}
	\item $J = 0^*$;
	\item $(V_j)_{j\in J}$ is a decreasing sequence, i.e. $V_i\supseteq V_j$ for $i\prec j$;
	\item A length of an word contained in $V_i$ is at least $\vert i \vert+1$, i.e. $x\in V_i\Rightarrow \vert x \vert > \vert i \vert$.
\end{itemize}
\end{theorem}
\begin{proof}
Firstly, we construct an automatic family $\mathcal{W} = (W_j)_{j\in J}$ as an intermediate step. Let $J=0^*$ and we construct each member of the automatic family as follows:
\[W_j = \{x\mid \forall i[\vert i \vert \le \vert j \vert\Rightarrow \exists y\in U_i(y\prec x)] \}
\]
Due to the first-order definability property of automatic structures, $(W_j)_{j\in J}$ is a proper automatic family. From the definition it is clear that $(W_j)_{j\in J}$ is a decreasing sequence of languages. As for measure-theoretic properties, we have $\mu[W_j]\le \mu[U_i]$ for all $i$ such that $\vert i \vert \le \vert j \vert$. This shows that $\lim_{j\in J} \mu[W_j] = 0$, which makes $\mathcal{W}$ a proper ART. It is clear from construction that $\mathcal{U}$ and $\mathcal{W}$ are equivalent. So far, we have satisfied all of the desired conditions, except for the condition on a length of words. To address this condition, let us consider a following function, $\phi:0^* \to \Sigma^*$:
\[
\phi(j) = \min_{ll}(W_j) = \{x\mid x\in W_j,\, \forall y\in W_j(x\le_{ll}y)\}
\]
where $ll$ refers to \emph{length-lexicographic} linear order(sometimes called \emph{shortlex}) defined as follows:
\[
x\le_{ll} y \Leftrightarrow \vert x \vert < \vert y \vert \text{ or } \vert x \vert = \vert y \vert,\, x\le_{lex}y
\]
where $lex$ refers to usual lexicographic order. Both $lex$ and $ll$ are automatic relations. Hence, the function $\phi$ is automatic, due to the first-order definability property. Being an automatic function, $\phi$ has some pumping constant $c$. We are going to argue that any element of $W_j$ is no shorter than $\vert j \vert - c$ for any $j\in 0^*$. Otherwise, we have a pair $(x,j)$ such that $\phi(j)=x$ and $\vert j \vert > \vert x \vert + c$. Applying the pumping lemma for this instance, we assert an existence of arbitraily long $i$ such that $\phi(i)=x$. However, it contradicts the fact that $\lim_{j\in J}[W_j] = 0$. Hence, following condition on a length of words holds:
\[
x\in W_j \Rightarrow \vert x \vert \ge \vert j \vert - c
\]
Finally, we get rid of the constant $c$ in order to construct desired ART:
\[
V_j = W_{j0^{c+1}}
\]
Observe that the satisfaction of all other desired conditions remains intact by this shift in indexing. This completes our construction. As a final remark, let us observe that the conclusion of the theorem holds if we replace ART with MART, for the resulting randomness test preserves MART property.
\end{proof}

\paragraph*{Universal randomness tests}In the original theory of Martin-L\"{o}f radnomness(MLR) there is a notion of a \emph{universal ML test}. This test is expected to subsume all other Martin-L\"{o}f tests, so that testing a sequence against an universal ML test reveals its randomness properties. In our notations, universal ART Test $\mathcal{V}$ is expected to satisfy the following property: for any ART $\mathcal{U}$, we have that  $F(\mathcal{U})\subseteq F(\mathcal{V})$. We are going to show that there is no such universal test in our case. Firstly we verify some facts.

\begin{lemma}
Suppose we are given a finite automaton $M=(S,f,s_0,F)$. Assume that for every state, $q\in S$, there is a word, $u$, such that $q\cdot u\in F$. Then a measure of all sequences which never visit an accepting state when run on $M$ is zero.
\end{lemma}
\begin{proof}
Let $\{s_0,s_1,\ldots,s_n\}$ be states of $M$. We construct a word $w$, processing which from any state results in a visit to some accepting state. The desired word $w=w_0w_1\ldots w_n$ is constructed inductively. Choose $w_0$ such as $s_0\cdot w_0\in F$. For any $i>0$, choose $w_i$ such that $(s_iw_0w_1\ldots w_{i-1})\cdot w_i\in F$. In this way final word $w$ satisfies given requirements. Let us denote the class of sequences which never visits an accepting state as $\mathcal{C}$. Given any $X\in \mathcal{C}$, it is clear that $X$ does not contain $w$ as a subword. Thus, $X$ is not disjunctive. This implies $\mathcal{C}$ is a subcollection of nondisjunctive sequences. Since the latter has measure zero, the former should also have measure zero. 
\end{proof}

\begin{theorem}[No universal test for ART]
There is no universal ART. In other words, there is no ART $\mathcal{V}$ such that for any ART $\mathcal{U}$ we have $F(\mathcal{U})\subseteq F(\mathcal{V})$
\end{theorem}
\begin{proof}
Let us assume contrary, i.e. there is an universal randomness test $\mathcal{V}=(V_j)_{j\in J}$. It implies that for an arbitrary ART $\mathcal{U}=(U_i)_{i\in I}$we have:
\[
F(\mathcal{U})\subseteq F(\mathcal{V})
\]
Since $\liminf_j \mu[V_j]=0$, there is an index $j\in J$ such that $\mu[V_j]<1$. We wish to show that the class $[V_j]$ is not dense, i.e. there is some string $w$ which cannot be extended to an element of $[V_j]$. If this is a case, one can easily construct an ART $\mathcal{U}=(U_i)_{i\in I}$ so as to achieve a contradiction:
\[
U_i=\{wi\}
\]
where $I=0^*$. Clearly, an infinite sequence $w0^{\omega}\in F(\mathcal{U})$, yet $w0^{\omega}\not\in F(\mathcal{V})$ due to the fact that $w0^\omega\not\in [V_j]$. To complete the proof, we are left to show that $[V_j]$ is not a dense class. Let us observe that being a member of automatic family, $V_j$ is a regular language on its own, with some underlying finite automaton $M$. Since $[V_j]$ is an extension of all strings in $V_j$, it can be viewed a collection of all infinite sequences visiting some accepting state of $M$. To show that $[V_j]$ is not dense, it suffices to show an existence of some state, $q$ in $M$ such that for any word $w$, we have that $q\cdot w\not\in F$. If it is not the case, then according to Lemma 1, we have $\mu[V_j]=1$. As we are given that $\mu[V_j]<1$, there must be such a state $q$. Let $w$ be a word such that $s_0 \cdot w=q$. Then $w$ exhibits the fact that $[V_j]$ is not dense. Together with earlier argument, this completes the proof.  
\end{proof}

\subsection{Machine version of ART}
One of the common practices in automata theory is a search for machine characterization of a particular phenomenon. In our case, we are interested in a covering of an infinite sequence $X$ by some ART $\mathcal{U}$. A natural question to ask: if this can be checked on some machine model. Since $X$ is an infinite sequence, we need to deal with $\omega$-automata. It turns out that for given ART $\mathcal{U}$, its covering region $F(\mathcal{U})$ can be described as a language corresponding to some deterministic B\"{u}chi automata of measure zero. Moreover, converse statement also holds. 

\begin{theorem}[Characterization of a covering region]
Let $\mathcal{U}$ be an ART. Then $F(\mathcal{U})$ is recognized by a deterministic B\"{u}chi automaton of measure zero.
\end{theorem}
\noindent
We are going to provide two genuinely different proofs. The first proof is based on a direct construction of a desired B\"{u}chi automata. The second proof is based on an indirect argument, where we do not know how a corresponding automaton is going to look like. 
\begin{proof}[Direct Construction]
According to Theorem 1, we can assume that $\mathcal{U}=(U_i)_{i\in I}$, where $I=0^*$ and the length of words in $U_i$ is larger than that of $i$ for all $i\in I$. Our aim is to construct a deterministic B\"{u}chi automaton $M$ such that $F(\mathcal{U})=L(M)$. The final automaton $M$ is based on deterministic movements of multiple markers inside of a some automaton. The idea is similar to the one of constructing a deterministic finite automaton equivalent to a given nondeterministic finite automaton. Recall that $X\in F(\mathcal{U})$ if and only if for all $i\in I$, there is $x_i\in U_i$ such that $x_i\prec X$. So, $M$ should somehow verify if above condition is satisfied while reading elements of $X$. The idea is to check for an existence of $x_i$ for each $i$ one-by-one. Suppose we want to check existence of $x_i$ for $i=0^n$. Let $N$ be an automaton recognizing automatic relation corresponding to $\mathcal{U}$. First we feed a following block into $N$:
\[
\begin{pmatrix}
X[n]\\
0^n
\end{pmatrix}
\]
Then we associate a red marker with this process to keep track a state of $N$, when subsequent blocks in the form of $(X_i,\#)$ are fed. If the given red marker is ever to visit an accepting state of $N$, then it is bound to disappear. Observe that the condition on existence of $x_i$ is now replaced by the condition on a disappearance of the created red marker. Now this procedure can be performed for all $n$ as we read elements $X$ bit-by-bit. The global condition, $X\in F(\mathcal{U})$, can be formulated as eventual disappearance of all created red markers. This procedure poses two kinds of problems:
\begin{itemize}
	\item Number of red markers might grow indefinitely
	\item As new red markers are being constantly produced, we might fail to keep track of the disappearance of older red markers. 
\end{itemize}
To address the first issue, we are going to merge all red markers corresponding to the same state. This ensures that at any given stage, we have at most $\vert N \vert$, number of states in automaton corresponding to $N$, many red markers. In order to address the second issue, we introduce a notion of old and new red markers. We refer to old red markers as grey markers. Then we divide overall dynamics into phases. At the beginning of each phase, all existing red markers are transformed into grey markers. Grey markers behave just as red markers. The end of each phase is characterized by the disappearance of all grey markers, after which the next phase begins. Observe that, the disappearance of all created red markers is equivalent to witnessing ends of mentioned phases infinitely often. Thus, if one sets all accepting states to the end of phases and assigns states for each possible configuration of markers, then we have a condition corresponding to an acceptance condition for the final B\"{u}chi automaton. This completes the construction.
\end{proof}

\begin{proof}[Indirect Construction]
Given an ART, $\mathcal{U}=(U_i)_{i\in I}$, satisfying properties given in Theorem 1, we wish to build a deterministic B\"{u}chi automaton of measure zero recognizing the language $F(\mathcal{U})$. Since $F(\mathcal{U})$ is of measure zero, it suffices to show that $F(\mathcal{U})$ is recognized by some deterministic B\"{u}chi automaton. Recall that given class $\mathcal{C}$ is recognizable a deterministic B\"{u}chi automaton if there is a regular language $R$ satisfying:
\[
X\in \mathcal{C} \Leftrightarrow Pref(X)\cap R\text{ is an infinite set}
\] 
So we need to exhibit a regular language $R$ satisfying the above condition for $F(\mathcal{U})$. Recall that $X\in F(\mathcal{U})$ iff for all $i\in I$, there is $x_i\in U_i$ such that $x_i\prec X$. Since $(U_i)_{i\in I}$ is assumed to be a decreasing sequence, we have $X\in F(\mathcal{U})$ if and only if there is infinitely many $i\in I$ such that the mentioned property holds, i.e. $\exists^{\infty} i\in I$ such that there is $x_i\in U_i$ with $x_i\prec X$. Then we turn each $U_i$ into a prefix-free language, which gives us a new automatic family $\mathcal{V} = (V_i)_{i\in I}$. Formally,
\[
V_i = \{x\in U_i: \forall y\in U_i(y\not\prec x)\}
\]
We can easily verify that $X\in F(\mathcal{U})$ if and only if $\exists^{\infty}i\in I$ such that there is $x_i\in V_i$ with $x_i\prec X$. As for a forward direction, given $X\in F(\mathcal{U})$, we have $\exists^{\infty}i\in I$ there is $x_i\in U_i$ such that $x_i\prec X$. For each mentioned $i\in I$, we select $y_i\preceq x_i\prec X$ such that $y_i\in V_i$. On the other hand, suppose that $\exists^{\infty}i\in I$ such that there is $x_i\in V_i$ with $x_i\in X$. By a virtue of $V_i$ being a subset of $U_i$, we have that $\exists^{\infty}i\in I$ such that there is $x_i\in U_i$ with $x_i\prec X$. Finally, let us define the desired regular language $R$:
\[
R = \bigcup_{i\in I}V_i = \{x\mid \exists i(x\in V_i)\}
\]
Due to the first-order definability for automatic structures, $R$ is indeed a regular language. Let us show that it satisfies the desired properties. Suppose that $X\in F(\mathcal{U})$, then we know $\exists^{\infty}i\in I$ such that there is $x_i\in V_i$ with $x_i\prec X$. Since the initial $\mathcal{U}$ is assumed to satisfy the condition on the length of words, we have
\[
x\in V_i \Rightarrow \vert x \vert > \vert i \vert
\]
Together with former observation, we conclude that $R\cap Pref(X)$ is indeed an infinite set. On the other hand, suppose that $R\cap Pref(X)$ is infinite. Since each $V_i$ is a prefix-free language, no two elements of $Pref(X)$ can belong to the same $V_i$. Hence, there are infinitely many $i\in I$ such that there is $x_i\in V_i$ with $x_i\prec X$. This implies that $X\in F(U)$. Thus, $R$ indeed satisfies the desired properties. 
\end{proof}
\noindent
Now we state a converse direction of the previous theorem.
\begin{theorem}
Let $M$ be a deterministic B\"{u}chi automaton of measure zero. Then there is ART $\mathcal{U}$ such that $L(M)=F(\mathcal{U})$.
\end{theorem}

\begin{proof}
Let $R$ be a language corresponding to $M$, when it is viewed as a finite automaton, instead of the determinsitic B\"{u}chi automaton. Set $I=0^*$ and define $\mathcal{U} = (U_i)_{i\in I}$ as follows:
\[
U_i = \{x\mid \vert x \vert \ge \vert i \vert,\, x\in R\}
\]
Due to the first-order definability of automatic structures $\mathcal{U}$ is a proper automatic family. Moreover $\mathcal{U}$ is a decreasing family in a sense that $U_i\supseteq U_j$ for $i\prec j$. Recall that $X\in L(M)$ if and only if $Pref(X)\cap R$ is an infinite set. The latter is equivalent to saying that $X\in [U_i]$ for all $i\in I$, i.e. $X\in F(\mathcal{U})$. This argument shows that $L(M)=F(\mathcal{U})=\bigcap_{i\in I}[U_i]$. Since a measure of the Cantor space itself is bounded, we apply Lebesgue dominated convergence theorem to the previous equation, which gives us:
\[
\mu(L(M)) = \lim_{i\in I} \mu[U_i]
\]
As $\mu(L(M))=0$, we have that $\mathcal{U}$ forms a proper ART which satisfies the desired properties. 
\end{proof}

\begin{corollary}
Let $\mathcal{C}$ be a class of the Cantor Space. The following are equivalent:
\begin{enumerate}
	\item There is ART $\mathcal{U}$ such that $\mathcal{C}=F(\mathcal{U})$.
	\item There is a deterministic B\"{u}chi automaton of measure zero such that $\mathcal{C}=L(M)$
\end{enumerate}
\end{corollary}
\noindent
We have encountered deterministic B\"{u}chi automata of measure zero. A natural question to ask would be: when is a given deterministic B\"{u}chi automaton of measure zero. To answer this question, we need to verify few things given in the form of lemmas.
We should note that some of the following results also appear in the work of Staiger \cite{Null,Monadic}, but we wish to give our own proofs. 

\begin{lemma}
Let $M=(S,f,s_0)$ be an FSM. Then a run of any disjunctive sequence $X$ is going to reach a leaf connectied component, i.e. $\exists i\in \mathbb{N}$ such that $s_0\cdot X[i]\in g$ for some $g\in (M)$. 
\end{lemma}
\begin{proof}
Firstly, let us observe that for any state $q\in S$ there is a word $w$ such that $q\cdot w$ is in some leaf connected component. Following an idea in the proof of the lemma 1, it is possible to construct a word $w$ such that no matter from which state one starts with, processing $w$ brings one to a leaf connected component. Given a disjunctive sequence $X$, it contains $w$ as a subword by definition. Hence, a run of disjunctive sequence $X$ on $M$ is going to end up in some leaf connected component. 
\end{proof}
\begin{lemma}
Given a disjunctive sequence $X$, any suffix of $X$ is also a disjunctive sequence. 
\end{lemma}
\begin{proof}
Suppose that we want to prove that $X[i:]=X_{i+1}X_{i+2}\ldots$ is disjunctive for some $i$. Let us pick some word $w$. In order to show that $w$ appears as a subword in $X[i:]$, let us observe that $0^iw$ appears as a subword in $X$. Since no part of $w$ could appear in $X[i]=X_1X_2\ldots X_i$, we have that $w$ appears as a subword in $X[i:]$. 
\end{proof}

\begin{theorem}[Disjunctive sequences and automata]
Let $M=(S,f,s_0)$ be an FSM. Then for any disjunctive sequence $X$, we have that $I(X)=g$ for some $g\in (M)$. 
\end{theorem}

\begin{proof}
According to Lemma 2 processing any disjunctive sequence $X$ will lead the automaton $M$ to some leaf connected component. Let $g\in (M)$ to be a leaf connected component in which $X$ ends up to be. Clearly $I(X)\subseteq g$, so we we need to show that $g\subseteq I(X)$. For an arbitrary state $q\in g$, we can apply the argument used in the proof of lemma 1 in the context of an FSM $g$ and an accepting state $q\in g$. There is a word $w_q$ such that processing $w_q$ starting from any state of $g$ would visit $q$ at least once. Now, if $I(X)\neq g$, then there should be a state $q\in g$ such that $q$ is never visited from some point onwards during the run of $X$. This implies that for some $i\in \mathbb{N}$, $X[i:]=X_{i+1}X_{i+2}\ldots$ does not contain $w_q$ as a subword. Thus, $X[i:]=X_{i+1}X_{i+2}\ldots$ is not disjunctive, which contradicts to our initial assumption according to Lemma 3. 
\end{proof}
\noindent
Finally, we provide a condition of a deterministic B\"{u}chi automaton to be of measure zero.
\begin{theorem}[Measure and B\"{u}chi automata]
Let $M=(S,f,s_0,F)$ be a deterministic B\"{u}chi automaton. Then $\mu(L(M))=0$ if and only if $F\cap g=\emptyset$ for all leaf connected components $g\in (M)$.
\end{theorem}
\begin{proof}
Forward direction $(\Rightarrow):$\\
Assume that $q\in F\cap g$ for some leaf cpnnected component $g$. Since we assume that all states are reachable from starting state, $s_0$, there is a word $w$ such that $s_0\cdot w=q$. Now consider any disjunctive sequence $X\in\mathcal{D}$. By applying Theorem 5 for a component $g$ and a starting state $q$, we infer that $wX$ visits every element of $g$ infinitely often. This means that $wX$ is going to be an accepting sequence. Since the measure of all disjunctive sequences is one, we have that $\mu(w\mathcal{D})=2^{-\vert w \vert}$. Hence $\mu(L(M))\ge 2^{-\vert w \vert}>0$. Thus, the forward direction follows.\\
Backward direction $(\Leftarrow):$\\
Assume $F\cap g=\emptyset$ for all leaf components $g\in (M)$. According to Theorem 5, for any disjunctive $X$, $I(X)=g$ for some leaf component $g\in (M)$. This means that any sequence accepted by $M$ is not disjunctive, and $L(M)$ is a subcollection of nondisjunctive sequences. As a latter has the measure of zero. Former should also have the measure of zero. This completes the backward direction, and the overall proof of the given theorem.
\end{proof}
\noindent
Since deterministic Muller automata are close to deterministic B\"{u}chi automata, we provide same type of condition for deterministic Muller automata. Later this observation turns out to be useful for us. 
\begin{theorem}[Measure and Muller Automata]
Let $M=(S,f,s_0,F)$ be a deterministic Muller automaton. Then $\mu(L(M))=0$ if and only if $F\cap (M) = \emptyset$. 
\end{theorem}
\begin{proof}
The proof is similar to the one of above theorem. Forward direction\\$(\Rightarrow):$\\
Again assume that $g\in F$ for some leaf component $g\in(M)$. Let us choose an arbitrary state, $q\in g$. By an initial assumption on automata, there is a word $w$ such that $s_0\cdot w=q$. Then applying Theorem 5 for a connected component $g$ with a starting state $q$, we observe that for $wX$ where $X\in \mathcal{D}$, we have $I(wX)=g$. Since the measure of disjuncitve sequences is one, $\mu(w\mathcal{D})=2^{-\vert w \vert}$. This implies that $\mu(L(M))\ge 2^{-\vert w \vert}>0$. This completes the forward direction.\\
Backward direction, $(\Leftarrow):$\\
Assume $g\not\in F$ for any leaf connected component $g\in(M)$. According to the Theorem 5, for any disjunctive sequence $X$, we have $I(X)=g$ for some leaf component $g\in(M)$. For any sequence $X$, accepted by $M$, we have that $X$ is not disjunctive. Hence $L(M)$ is a subcollection of nondisjunctive sequences. As a latter has the measure zero, former should also have the measure zero. This completes the backward direction, and the whole proof of the given theorem. 
\end{proof}

\subsection{Characterization}
\noindent
In this section we summarize our observations in the form of a characterization result for sequences covered by an ART. Again, note that equivalence of last three statements is known from Staiger \cite{Null,Monadic}
\begin{theorem}[Characterization]
Given an infinite sequence $X$ the following are equivalent:
\begin{enumerate}
	\item $X\in F(\mathcal{U})$ for some ART $\mathcal{U}$
	\item $X$ is accepted by a deterministic B\"{u}chi automaton of measure zero.
	\item $X$ is accepted by a deterministic Muller automaton of measure zero.
	\item $X$ is not a disjunctive sequence
\end{enumerate}
\end{theorem}
\begin{proof}
$(1\Rightarrow 2):$\\
Suppose a sequence $X$ is covered by an ART $\mathcal{U}$. By Theorem 1, $F(\mathcal{U})$ is recognized by a deterministic B\"{u}chi automaton $M$ of measure zero. So clearly, $X$ is accepted by $M$, which has desired properties.\\
$(2\Rightarrow 3):$\\
Suppose that a sequence $X$ is accepted by a deterministic B\"{u}chi automaton $M$ of measure zero. It is known that nondeterministic B\"{u}chi automata is equivalent to deterministic Muller automata. As any deterministic B\"{u}chi automaton can be viewed as a nondeterministic B\"{u}chi automaton, every language recognized by a deterministic B\"{u}chi automaton has a corresponding determinsitic Muller automaton. So, there is a deterministic Muller automaton $N$ such that $L(M)=L(N)$. As $X\in L(N)$, we have that $X$ is accepted by $N$.\\
$(3\Rightarrow 4):$\\
Suppose that a sequence $X$ is accepted by a deterministic Muller automaton $M$ with accepting collection $F$. According to Theorem 7, we have $F\cap(M)=\emptyset$. Combining this observation with Theorem 5, we obtain that any sequence accepted by $M$ is nondisjunctive. In particular, $X$ is not disjunctive. \\
$(4 \Rightarrow 1):$\\
Suppose that $X$ is nondisjunctive, i.e. there is some word $w$ which does not appear in $X$ as a subword. To construct ART covering $X$, let us take:
\[U_i = \{x\mid \vert x \vert = \vert i \vert, \, x \text{ does not contain }w\text{ as a subword}\}
\]
with $I=0^*$. Clearly, $\mathcal{U}=(U_i)_{i\in I}$ is an automatic family. Furthermore, a measure of $\mu[U_i]$ could be bounded in a straightforward manner. Assuming that $\vert w \vert = d$, consider an element  $u\in U_{dk}$ for some $k\in \mathbb{N}$. As $u$ can be divided into $k$ blocks of size $d$, and none of the block is allowed to be $w$, we have:
\[\mu[U_{0^{dk}}]\le \left(1-\frac{1}{2^d}\right)^k
\]
As the limiting value of the right hand side approaches zero as $k\to \infty$, we have $\liminf_{i\in I} \mu[U_i]=0$. This shows that $(U_i)_{i\in I}$ is a proper ART. Since $X\in [U_i]$ for all $i\in I$, we have $X\in F(\mathcal{U})$, i.e. $X$ is covered by $\mathcal{U}$. 
\end{proof}

\begin{theorem}[Combinatorial characterization of AR]
An infinite sequence $X$ is automatic random~(AR) if and only if it is disjunctive.
\end{theorem}
\begin{proof}
Focusing on $(1)$ and $(4)$ of the above theorem~(Theorem 8), we are aware of the following equivalence:
\[X \text{ is covered by some ART }\Leftrightarrow X\text{ is not disjunctive }
\]
This is equivalent to saying:
\[X \text{ is AR}\Leftrightarrow X \text{ is disjunctive } 
\]
This way, we obtained a complete combinatorial equivalent of AR condition.
\end{proof}

\subsection{On relation between ART and MART}
At last, we would like to say few words regarding a relationship between ART and MART. Clearly ART subsumes MART, because each MART can be viewed as an ART. A natural questions to ask is an equivalence of these notions. Do these notions coincide or do they give rise to different classes of randomness? We are going to show that these two notions are not equivalent, i.e. there is a class $\mathcal{C}$ such that $\mathcal{C} = F(\mathcal{U})$ for some ART $\mathcal{U}$, yet $\mathcal{C} \neq F(\mathcal{V})$ for any MART $\mathcal{V}$. 
\begin{theorem}
ART and MART are not equivalent as randomness notions.
\end{theorem}
\begin{proof}
Let us consider a class presented as an example earlier in the paper:
\[
\mathcal{C} = \{X\mid X_{2i}=0\text{ for all } i\in \mathbb{N}\}
\]
We have shown that there is an ART $\mathcal{U} = (U_i)_{i\in I}$ given as:
\[
I=(00)^*,\quad U_i=(\Sigma 0)^{\frac{\vert i \vert}{2}}
\]
such that $\mathcal{C} = F(\mathcal{U})$. Observe that for each $x\in U_i$, we have $\vert x \vert = \vert i \vert$. Moreover $\vert U_i \vert = 2^{\frac{\vert i \vert}{2}}$. We need to show that $\mathcal{C} \neq F(\mathcal{V})$ for any MART $\mathcal{V}$. Assume contrary: there is a such MART $\mathcal{V}$. Let us apply Theorem 1, to transform $\mathcal{V}$ to more favorable MART $\mathcal{W} = (W_e)_{e\in E}$ satisfying:
\begin{itemize}
	\item $E = 0^*$.
	\item $(W_e)_{e\in E}$ is decreasing, i.e. $W_i\supseteq W_j$ for $i\prec j$.
	\item The length of words contained in $W_e$ is larger than $\vert e \vert$, i.e. $x\in W_e\Rightarrow \vert x \vert > \vert e \vert$.
\end{itemize}
Let us consider $W_{0^{2k}}$ for some $k$. For the sake of notational convenience, we are going to write $kk$ in place of $0^{2k}$. Any $X\in \mathcal{C}$ should have some prefix $x_{2k}\prec X$ in $W_{kk}$, satisfying $\vert x_{2k}\vert >  2k$. This means that any $x\in U_{kk}$ has some extension $y$ in $W_{kk}$. Let $c$ a pumping constant corresponding to the automatic relation for $\mathcal{W}$. For any string $y\in W_{kk}$ such that $\vert y\vert > 2k + c$, we can apply the pumming lemma to pump down $y$ to $y'$ such that $\vert y' \vert \le 2k + c$. Hence, for any $x\in U_{kk}$ we can assume that its extension $y$ has a length of at most $2k+c$. Finally, we estimate the measure of $[W_{kk}]$:
\[
\mu[W_{kk}]\ge \sum_{x\in U_{kk}}2^{-\vert y \vert} \ge 2^k(2^{-2k-c})=2^{-k-c}> 2^{-2k}
\]
provided that $k> c$. This fact contradicts the requirement of $\mathcal{W}$ being a MART. 
\end{proof}
\noindent
So far, we have seen a difference between ART and MART on a class level, i.e. there is a class $\mathcal{C}\subseteq \{0,1\}^{\mathbb{N}}$ which is covered by some ART and none of MART. It turns out the difference goes further into a sequence level. In particular, we are going to show an existence of some sequence $X$ which is random with respect to MART, yet not random with respect to ART.
\begin{theorem}
There is a sequence $X$ which is Martin-L\"{o}f automatic random~(MAR), yet it is not automatic random~(AR).
\end{theorem}
\begin{proof}
Let $A$ be some disjunctive sequence and $X=A\oplus 0^{\omega}=A_10A_20\ldots$. Clearly, $X$ is not AR, for there is ART covering $X$ as shown in Examples subsection of the current paper. We are left to show that $X$ is MAR, i.e. no MART covers $X$. Assume contrary, suppose there is some MART $\mathcal{U}=(U_i)_{i\in I}$ covering $X$, which can be assumed to satisfy properties given in Theorem 1. Let us restrict our attention to even indexed language families, $U_{0^{2k}}$. Again for the sake of notational convenience, we replace an index of $0^{2i}$ by $ii$. Let us consider an induced family of languages, $\mathcal{V}=(V_i)_{i\in I}$ with $I=0^*$:
\[
V_i = \{a_1a_2\ldots a_n\mid (a_10a_20\ldots a_n\in U_{ii})\text{ or }(a_10a_20\ldots a_n0\in U_{ii})\}
\]
It is possible to construct a finite automaton recognizing $\mathcal{V}$ from that recognizing $\mathcal{U}$ by 'skipping' $0$ transitions. This argument shows that $\mathcal{V}$ is actually an automatic family. Furthermore, we have that $A\in \bigcap[V_i]$. This points out that $\mathcal{V}$ is not ART, for $A$ is a disjunctive sequence. In order to understand $\mathcal{V}$ better, let us switch to the realm of B\"{u}chi automata. As we know, there is some B\"{u}chi automaton $M$ recognizing $\bigcap[U_{ii}]$. A modification of $M$ consisting of 'skipping' over $0$ transitions, gives as a B\"{u}chi automaton $N$ recognizing $\bigcap[V_i]$. As for $N$, we know following facts: it is of positive measure and $A\in L(N)$. The latter fact can be interpreted as a presence of an accepting state in a connected component where the run of $A$ ends up in. Thus there is a prefix $w\prec A$ such that for any disjunctive sequence $X\in \mathcal{D}$, we have that $wX\in L(N)$, equivalently $wX\in [V_i]$ for all $i$. Now we can use this observation to derive a lower bound for the measure of $U_{ii}$, thus achieving a contradiction. Given any string $v$ of length $\vert i \vert$ such that $w\preceq v$, there is an extension of $v$ in $V_i$. Translating this statement for $U_{ii}$, we have: given $a_1a_2\ldots a_{\vert i \vert} \succeq w$, there is some string $u\in U_{ii}$ extending $a_10a_20\ldots a_n$. Thanks to pumping lemma, a length of $u$ can be assumed to be no greater than $2\vert i \vert + c$, where $c$ is a pumping constant correspong to automatic relation $\mathcal{U}$. Just as in previous proof, we now estimate $\mu[U_{ii}]$:
\[
\mu[U_{ii}]\ge \sum_{v}2^{-\vert u \vert} \ge 2^{\vert i \vert - \vert w \vert}2^{-2\vert i \vert - c} = 2^{-\vert i \vert - \vert w \vert - c} > 2^{-2\vert i \vert}
\]
given large enough $\vert i \vert$. This clearly violates conditions for being MART, which indicates falsity of the initial assumption. Hence, $\mathcal{U}$ could not be MART.
\end{proof}
\noindent
So far, we have seen that it is impossible, in general, to impose a condition of $\mu[U_i]\le 2^{-\vert i \vert}$ on ART without altering its randomness properties. A question now is what can we say about individual measures at all? Can we impose some weaker conditions? It turns out individual measures can be assumed to decrease exponentially.
\begin{theorem}
Suppose that $\mathcal{C}=F(\mathcal{U})$ for some ART $\mathcal{U}$. Then there is an  ART $\mathcal{V} = (V_j)_{j\in J}$ subsuming $\mathcal{U}$, such that $\mu[V_j]\le \gamma^{\vert j \vert}$ for some $\gamma<1$.
\end{theorem}
\begin{proof}
Invoking a machine characterization of ART, there is a deterministic B\"{u}chi automaton of measure zero, $M=(S,f,s_0,F)$, recognizing $\mathcal{C}$. By Theorem 6, none of accepting states of $M$ is in a leaf connected component. Applying the argument used in the proof of Lemma 1, there is a word $w$ which brings any state of $M$ into a leaf connected component. This means that for any $X\in \mathcal{C}$, $X$ does not contain $w$ as a subword. Let $\mathcal{C}_w$ be a collection of all sequences not having $w$ as a subword. By the previous argument $\mathcal{C}\subseteq \mathcal{C}_w$. Recall that in the proof of Theorem 9, we have constructed an ART, $\mathcal{W}=(W_i)_{i\in I}$, corresponding to $\mathcal{C}_w$:
\[
W_i = \{x\mid \vert x \vert = \vert i \vert, \, x \text{ does not contain }w\text{ as a subword}\}
\]
We have shown that $\mathcal{W}$ is indeed an ART such that $\mathcal{C}_w = F(\mathcal{W})$. Furthermore, we have shown that:
\[
\mu[W_{0^{dk}}]\le \left(1-\frac{1}{2^d}\right)^k
\]
where $d=\vert w \vert$. Observe that $([W_j])_{j\in J}$ is a decreasing sequence, in a sense that $[W_i]\supseteq [W_j]$ for $i\prec j$. Hence, a condition $X\in F(\mathcal{W})$ can be stated as $\exists^{\infty}i\in I$ with $x_i\in W_i$ and $x_i\prec X$. Thus, choosing any infinite collection of members from $\mathcal{W}$ results in the exact same covering region. We only need to ensure that this choice somehow selects a regular language of indices. Collection $J = (0^d)^+ = 0^d(0^d)^*$ forms a regular language. So we set a new ART $\mathcal{V} = (V_j)_{j\in J}$ so that $V_j = W_j$. Given $j\in J$ such that $\vert j \vert = dk$, we have:
\[
\mu[V_j] \le (1-\frac{1}{2^d})^{k} = \gamma^{\vert j \vert}, \text{ where } \gamma = \left(1 - \frac{1}{2^d}\right)^{\frac{1}{d}}
\] 
This completes the construction. 
\end{proof}

\section{Discussions}
We have defined and investigated properties of randomness tests in the context of automata theory. These investigations led to quite unexpected connections between randomness tests and $\omega$-automata such as by B\"{u}chi and Muller. Furthermore, a purely combinatorial characterization of automatic randomness in the form of disjunctive property for sequences has been found. Let us compare our results with automatic randomness notions arising from other paradigms. As for the unpredictability paradigm considered in \cite{Stimm,OConnor}, automatic random sequences correspond to normal sequences as in \cite{Borel}. Similarly, for the incompressibility paradigm considered in \cite{Becher,Shen} random sequences correspond to normal ones. Clearly, a normality is much stronger property than being disjunctive. Hence, the current automatic randomness tests result in much larrger class of random sequences. It is an open question if one could modify the definition of automatic randomness tests, so that resulting randomness class coincides with that of normal sequences. Moreover,  we hope that many more notions from the theory of algorithmic randomness could find their counterparts in the context of automata theory. 

\section*{Acknowledgement}
We would like to thank Frank Stephan for numerous discussions held in person as well as via e-mail. 

\medskip
\bibliographystyle{unsrt}
\bibliography{Reference.bib}

\end{document}